\documentclass[a4paper,onecolumn,11pt,unpublished,dvipsnames]{quantumarticle}
\pdfoutput=1
\usepackage[utf8]{inputenc}
\usepackage[english]{babel}
\usepackage[T1]{fontenc}
\usepackage{algorithm, algpseudocode}
\usepackage{amsmath, amssymb, amsthm}
\usepackage{caption, subcaption}
\usepackage{dsfont}
\usepackage{physics}
\usepackage{tikz}
\usepackage{tikz-network}
\usetikzlibrary{decorations.pathreplacing,calligraphy}
\usepackage{pgfplots}
\usepgfplotslibrary{external}
\usepackage{hyperref}
\usepackage[numbers,sort&compress]{natbib}
\usepackage{graphicx}
\usepackage{lipsum}
\usepackage{cleveref}
\usepackage{standalone}

\pgfplotsset{compat=1.17}

\newtheorem{theorem}{Theorem}

\theoremstyle{definition}
\newtheorem{definition}[theorem]{Definition}

\newtheorem{lemma}[theorem]{Lemma}

\theoremstyle{remark}
\newtheorem*{remark}{Remark}

\definecolor{mypurple}{RGB}{138, 40, 199}
\definecolor{mygreen}{RGB}{9, 150, 79}

\begin{document}

\title{Composable privacy of networked quantum sensing}

\author{Naomi R. Solomons}
\email{naomi.solomons@lip6.fr}
\orcid{0000-0002-3332-5560}
\affiliation{LIP6, CNRS, Sorbonne Université, 4 Place Jussieu, F-75005 Paris, France}
\author{Damian Markham}
\orcid{0000-0003-3111-7976}
\affiliation{LIP6, CNRS, Sorbonne Université, 4 Place Jussieu, F-75005 Paris, France}
\maketitle

\begin{abstract}
  Networks of sensors are a promising scheme to deliver the benefits of quantum technologies in coming years, offering enhanced precision and accuracy for distributed metrology through the use of large entangled states. Recent work has additionally explored the privacy of these schemes, meaning that local parameters can be kept secret while a joint function of these is estimated by the network. In this work, we use the abstract cryptography framework to relate the two proposed definitions of quasi-privacy, showing that both are composable, which enables the protocol to be securely included as a sub-routine to other schemes. We give an explicit example that estimating the mean of a set of parameters using GHZ states is composably fully secure.
\end{abstract}


\section{Introduction}

Distributed parameter estimation, the problem of estimating a global function of locally held parameters, is of significant interest in quantum metrology due to the potential advantage to be gained from the use of entanglement. Specifically, in this paper we will consider the use of multiparty entangled states (e.g. GHZ states), shared across a network of $n$ members, to estimate some linear function of parameters held by the nodes~\cite{proctor2018multiparameter, rubio2020quantum}. Quantum advantage arises due to the possible quadratic improvement in precision over classical or local methods, with applications such as time synchronisation~\cite{komar2014quantum} or medical imaging~\cite{eldredge2018optimal}. It is also a central use case proposed for the quantum internet~\cite{wehner2018}, due to the decentralised nature of the protocol and the minimal quantum requirements on end users.

As proposed in~\cite{shettell2022private}, alongside the metrological advantage of using quantum states, there is an inherent cryptographic guarantee. The goal here is that each local parameter encoded by an agent should remain private -- that is, unknown to all other agents -- whereas the desired global function is revealed to every agent. More precisely, the definition of privacy that we consider is that no dishonest member of the network, alone or in collaboration with other dishonest parties, should be able to learn any more information about the overall set of parameters than that which could be learned from knowledge of their own parameter(s) and the target function itself~\cite{bugalho2025private, hassani2025privacy,junior2025privacy,ho2024quantum}.

In this work, we consider privacy within the framework of abstract (or constructive) cryptography~\cite{maurer2011constructive}. This paradigm enables us to show composable security, meaning that the protocol can be composed with other secure protocols, used multiple times, or used as part of a larger scheme. In comparison to the game-based framework, which considers specific attacks, we model the protocol as a \textit{resource} and show its interaction with other abstract systems, which provides a method to comprehensively analyse what information can be leaked by the protocol. Our work applies to existing privacy definitions, which can describe a variety of possible schemes, using different input states or quantum dynamics, and therefore is an important step towards securely implementing these schemes in the real world.

\paragraph{Our results.} The main results of this work consider the two previous methods of quantifying privacy considered in the literature. We show that the quasi-privacy definition, $\mathcal{P}(\mathcal{Q}, \mathbf{a})$, introduced in~\cite{bugalho2025private} is composably $\varepsilon$-secure, with $\varepsilon = \sqrt{1 - \mathcal{P}^2(\mathcal{Q}, \mathbf{a})}$ (Thm.~\ref{thm: luis defn composably private}). We then show that the quasi-privacy definition introduced in~\cite{hassani2025privacy} is composably $\varepsilon$-secure up to a constant factor (Thm.~\ref{thm: hassani defn composably private}). \\

The organisation of this paper is as follows: in Section~\ref{sec: background} we give a brief overview of networked parameter estimation with quantum sensors, and composable security proofs in abstract cryptography. In Section~\ref{sec: mean estimation}, we illustrate these by showing that the use of GHZ states to jointly estimate the mean of a set of parameters is composably secure. In Section~\ref{sec: general fn} we consider more general states and functions, and give a condition for quantum dynamics to be composably private, which we relate to the privacy definitions in~\cite{bugalho2025private} and~\cite{hassani2025privacy} in Sections~\ref{sec: luis defn} and~\ref{sec: majid defn} respectively. In Section~\ref{sec: state verification} we consider composition with state verification schemes. Together, these give a composably secure parameter estimation scheme which can be used even in networks with untrusted sources or channels, or dishonest parties.

\section{Background} \label{sec: background}

\subsection{Networked parameter estimation}

\begin{figure}[ht]
    \centering
    \includegraphics{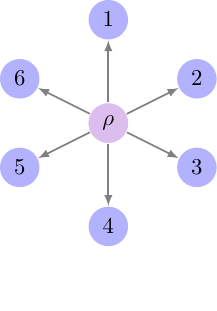}
    \hspace{0.2cm}
    \includegraphics{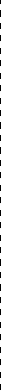}
    \hspace{0.2cm}
    \includegraphics{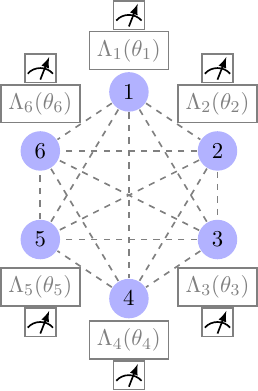}
    \hspace{0.2cm}
    \includegraphics{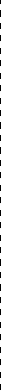}
    \hspace{0.2cm}
    \includegraphics{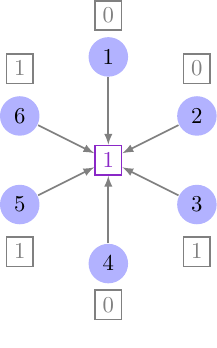}
    \caption{One round of the quantum parameter estimation protocol, adapted from~\cite{dejong2025anonymous}. A state $\rho$ is distributed across the network, then the parameters are encoded through the channels $\Lambda_\mu (\theta_\mu)$. Each party measures its own part of the state, and return the measurements to the rest of the network. The output at each round is collected and used to estimate $f(\boldsymbol{\theta})$.}
    \label{fig: parameter estimation}
\end{figure}

The basic protocol that we will consider is private quantum parameter estimation~\cite{shettell2022private,hassani2025privacy,bugalho2025private}. We consider an $n$-node network, where each member $\mu$ of the network holds a private local parameter $\theta_\mu$, and they wish to jointly estimate some linear function $f(\boldsymbol{\theta}) = \mathbf{a}\cdot\boldsymbol{\theta}$ (where we use $\boldsymbol{\theta}$ to represent the vector of $\{ \theta_\mu \}$). The quantum implementation is as follows (this is shown in Fig.~\ref{fig: parameter estimation}; see also~\cite{hassani2025privacy}, Fig. 1, or~\cite{bugalho2025private}, Fig. 2):
\begin{enumerate}
    \item A source distributes a state $\rho$ across the nodes of the network.
    \item Each node encodes their parameter by a quantum channel $\Lambda_\mu(\theta_\mu)$.
    \item Each node makes a measurement on their state and announces their outcome.
    \item Steps 1-3 are repeated $N$ times.
    \item Each node, or a co-ordinator, uses the measurement outcomes across all repetitions to make an estimate $\hat{f}(\boldsymbol{\theta})$ of the desired function.
\end{enumerate}

For example, in the case that $f(\boldsymbol{\theta})$ is the mean, $\Bar{\theta} = \frac{1}{n}\sum_\mu \theta_\mu$, this could be realised as follows:
\begin{itemize}
    \item The source distributes a GHZ state, $\frac{1}{\sqrt{2}}\left(\ket{0}^{\otimes n} + \ket{1}^{\otimes n}\right)$.
    \item The parameters are encoded through the unitary $\ketbra{0}{0} + \mathrm{e}^{i \theta_\mu} \ketbra{1}{1} = \mathrm{e}^{-i \theta_\mu \sigma_Z/2}$ (equality is up to a global phase).
    \item Each node measures in the $\sigma_X$ basis and announces the outcome, $o_\mu$.
    \item The probability that $\bigoplus_\mu o_\mu = 0$ is $\frac{1}{2}\left(1 + \cos(n\Bar{\theta})\right)$.
\end{itemize}

In general, the source may not be trusted, in which case a verification protocol can be used (such as~\cite{unnikrishnan2022verification, pappa2012multipartite}) to ensure that the state provided is sufficiently close to the required state (e.g. GHZ state) to guarantee an acceptable level of privacy. Other variations can be used according to further cryptographic requirements, such as one of the nodes keeping their outcomes secret or announcing a random bit in order to hide the function estimation from eavesdropping parties, or implementing further steps to carry out function estimation among an anonymous subset of participants~\cite{dejong2025anonymous}.

The privacy of this scheme has been considered in~\cite{shettell2022private, bugalho2025private, hassani2025privacy}. The key idea is that any member (or collaborating subset) of the network, or other adversary, whether or not they follow the protocol as intended, cannot access the individual parameters $\{ \theta_\mu \}$, or any function of these that cannot be calculated just from $f(\boldsymbol{\theta})$ and their own parameters. In the case that state verification is used, this can also protect against information leakage when adversaries are collaborating with the source. Further work includes different definitions of quasi-privacy when imperfect initial states are distributed, investigations of the effects of different sources of error, and classifying states which are useful for the private estimation of different functions. These proofs generally rely on the quantum Fisher information (QFI), a quantity used in quantum metrology to describe the information which can be extracted from a particular state~\cite{sidhu2020geometric}, to which we give an introduction in the appendix, Section~\ref{sec: QFI appendix}.

Privacy in distributed parameter estimation can also be achieved using classical resources~\cite{mo2016privacy, wang2023differentially}. Under sufficient cryptographic assumptions (such as sharing keys which are random to each user but sum to a known value) this can straightforwardly be achieved using local parameter estimation and secure multi-party computation. Alternatively, a large entangled state could be used to distribute shared correlations, which can then be used to encode announcements, thus obscuring individual parameters and only revealing the total parity (indeed, this is the mechanism underlying privacy in the quantum case). However, the central advantage of quantum resources in this scheme arises from enhanced precision; this is a quadratic improvement if the quantum Cramér-Rao bound is saturated (further details are given in appendix Section~\ref{sec: QFI appendix}). Therefore, the most pertinent uses of this scheme are those for which increased precision is required, without compromising the security (for example in industrial, medical or military scenarios).



\subsection{Composable security}

\begin{figure}[ht]
    \centering
    \includegraphics{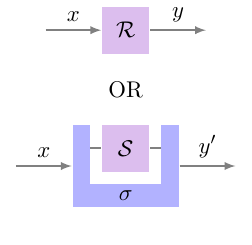}
    \hspace{0.6cm}
    \includegraphics{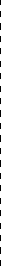}
    \hspace{0.6cm}
    \includegraphics{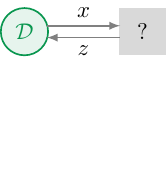}
    \caption{The interactions between systems in a constructive cryptographic proof. Given a concrete resource $\mathcal{R}$ and an ideal resource $\mathcal{S}$, a simulator $\sigma$ is constructed so that a distinguisher $\mathcal{D}$ cannot distinguish between $\sigma \mathcal{S}$ and $\mathcal{R}$. Roughly, this means that $\mathcal{R}$ is not revealing any information that is not already revealed by $\mathcal{S}$.}
    \label{fig: composable security}
\end{figure}

Composable security is a model for cryptographic security proofs that, by its construction, allows a protocol to be reused many times, or implemented as a sub-protocol in a more complicated procedure. Analysing protocols in this way is enabled by the frameworks of abstract, or constructive, cryptography~\cite{MauRen11, maurer2011constructive, portmann2014cryptographic} and universal composability~\cite{canetti2001universally}, with the former used in this work (this is in contrast with the game-based paradigm, which presents particular adversarial behaviour, and the resulting leakage of information). A useful pedagogical introduction is given in~\cite{colisson2024all}.

Composable security proofs are structured by introducing an ideal process, which specifies the desired functionality of the scheme, receiving inputs from its participants and returning appropriate outputs. The actual protocol can then securely realise the task if it is able to `emulate' the ideal process, in which case any outcomes caused by an adversary do not exceed what is inherently implied by the formal definition of the scheme. This is evaluated by considering the capability of some agent (the `distinguisher') to distinguish between the ideal and concrete implementations of the scheme. In full generality the distinguisher can act in collaboration with adversaries or corrupted protocol users, in any environment, in possession of any information that may be gained from repeated applications of the protocol, with any computational power and giving arbitrary inputs. However, the powers of the distinguisher can also be limited, making a slightly different security statement -- for example by only considering distinguishers with polynomial computational power. In this work, we make no assumptions on the limitations of the distinguisher in our proofs.

Abstract cryptography provides a formalism that is used to represent cryptographic definitions (such as programs or resources), so that certain technical details can be overlooked. Specifically, we describe systems, abstract objects with interfaces for the system to interact with its environment, which can be composed of various subsystems, and are categorised as resources, converters (including simulators or filters), or distinguishers -- an overview of these interactions is shown in Fig.~\ref{fig: composable security}. 

We use the following terminology and notation in our proof of composable security:
\begin{itemize}
    \item An ideal resource, $\mathcal{S}$, which represents the desired functionality of the protocol. A resource is a system with specified interfaces, accessible to particular users, from which it receives inputs and provides outputs. We will represent the set of interfaces with an honest user by $H$, and the set of interfaces with dishonest users (including eavesdroppers) by $D$.
    \item A concrete resource, $\mathcal{R}$, which, equipped with converters $\{\pi_i\}$ which describe the behaviour of users at the interfaces $\{ i \}$ of a resource, implement the protocol under investigation. So for example, $\pi_H \mathcal{R}$ is the resource $\mathcal{R}$, where the interfaces in the set $H$ follow the protocol $\pi$.
    \item A filter, $\Diamond$ or $\sharp$, which is a converter that may be applied to the interface of a resource to enforce honest behaviour (that is, to construct the resource as it would behave with no adversary present).
    \item A simulator, $\sigma$, which is a converter used for the purpose of the security proof.
\end{itemize}

The distinguishing advantage, with respect to a particular distinguisher $\mathcal{D}$ -- being an agent or group of agents, which may be given some restriction, such as being computationally bounded -- is defined as follows: if $\mathcal{D}$, interacting with the open interfaces of a system which may be either $\mathcal{A}$ or $\mathcal{B}$, can guess correctly with probability $p^{\mathcal{D}}(\mathcal{A}, \mathcal{B})$ with which system it is interacting, its advantage is:
\begin{equation}
    d^{\mathcal{D}}(\mathcal{A}, \mathcal{B}) := 2p^{\mathcal{D}}(\mathcal{A}, \mathcal{B}) - 1.
\end{equation}

This will vary according to the choice of distinguisher, so we drop the superscript $\mathcal{D}$ when we consider $d$ over all choices of computationally unbounded distinguisher, which leads to the notion of information-theoretic security. If $d(\mathcal{A}, \mathcal{B}) \leq \varepsilon$ for any choice of distinguisher, we write $\mathcal{A} \approx_{\varepsilon} \mathcal{B}$.

We are now in a position to introduce the main structure of constructive cryptography, adapted from~\cite{portmann2014cryptographic}:
\begin{definition} \label{defn: constructive cryptography}
    Let $\pi_H$ be a protocol and $\mathcal{R}_{\sharp}$ and $\mathcal{S}_{\Diamond}$ denote two filtered resources. We say that $\pi_H$ \textit{constructs} $\mathcal{S}_{\Diamond}$ \textit{from} $\mathcal{R}_{\sharp}$ \textit{within} $\varepsilon$, which we write $\mathcal{R}_\sharp \xrightarrow{\pi, \varepsilon} \mathcal{S}_\Diamond$, if the following two conditions hold:
    \begin{enumerate}
        \item Availability, or correctness:
        \begin{equation}
            \pi_H \mathcal{R} \sharp_D \approx_{\varepsilon} \Diamond_H \mathcal{S} \Diamond_D.
        \end{equation}
        \item Security: there exists a simulator $\sigma$ such that:
        \begin{equation}
            \pi_H \mathcal{R} \approx_{\varepsilon} \sigma_H \mathcal{S} \sigma_D.
        \end{equation}
        If it is clear what filtered resources $\mathcal{R}_{\sharp}$ and $\mathcal{S}_{\Diamond}$ are meant, we simply say that $\pi_H$ is $\varepsilon$-secure.
    \end{enumerate}
\end{definition}

In this definition, $\Diamond$/$\sharp$ is a filter which is applied to $\mathcal{S}$/$\mathcal{R}$ to enforce honest behaviour (note that $\Diamond_H \mathcal{S} \Diamond_D$ is the filtered resource, whereas $\mathcal{S}_\Diamond$ indicates that we will be using the resource $\mathcal{S}$ and applying the filter $\Diamond$ when appropriate for the proof). The subscript $H$/$D$ indicates that we apply these filters to the interfaces of $H$/$D$. If $\varepsilon = 0$ we omit the subscript and write e.g. $\pi_H \mathcal{R} \sharp_D \approx \Diamond_H \mathcal{S} \Diamond_D$.

This is shown to be a composable definition of security~\cite{maurer2011constructive}, hence this will be the method used in this work. Informally, we can understand this definition (for the case of $\varepsilon = 0$) the following way: $\mathcal{S}$ is an ideal protocol that we would like to implement, which is well-defined so that it is clear exactly how much information can be leaked to the environment or members of the scheme (for example, in distributed parameter estimation, we allow that the function of parameters can be made publicly available, but not the individual parameters). $\pi_H \mathcal{R} \pi_D$ is some realistic way of implementing this scheme, and we are interested in knowing whether it leaks any information that is not revealed by $\mathcal{S}$ if the members of $D$ follow a different behaviour. The simulator is a good way of representing this -- if we can use the information from $\mathcal{S}$ to construct all the information that an adversary would receive from $\pi_H \mathcal{R}$, then it must be the case that the open interfaces of $\pi_H \mathcal{R}$ (the information received by adversaries) do not contain any extra information. This is the security aspect of the proof.

The correctness aspect of the proof focuses on whether the concrete implementation can implement the desired functionality, assuming that it is used appropriately. However, as the realistic protocol may fall short of perfect behaviour, the filter that is applied to the ideal resource can reduce the power of $\mathcal{S}$ to what is expected from $\mathcal{R}$ (consider, for example, that in a composably secure state verification protocol, malicious adversaries can apply a `correction' to the state~\cite{colisson2024all}). In Definition~\ref{defn: constructive cryptography}, we represent this by including a filter $\Diamond_H$ on the honest interfaces of $\mathcal{S}$, which is generally omitted (including in~\cite{portmann2014cryptographic}). Similarly, we include a simulator $\sigma_H$ on the honest interfaces of $\mathcal{S}$ for the security proof. Therefore, it should be noted that we are considering a procedure which, although it does not leak any more information than $\mathcal{S}$, may not be as powerful (indeed, is only as powerful as $\Diamond_H \mathcal{S}$). 


\section{Estimation of mean} \label{sec: mean estimation}

We still start with a particular example of estimating the mean of the local parameters, and hence $\rho$ is a GHZ state, and $\mathbf{a} = \frac{1}{n}(1, ..., 1)^\mathrm{T}$. The ideal resource, $\mathcal{S}$, is then simply:
\begin{algorithm}
    \caption{Ideal resource for mean estimation, $\mathcal{S}$}\label{ideal resource mean}
\begin{enumerate}
    \item Receive the parameter $\theta_\mu \in [0, 2\pi)$ from each party.
    \item Return $\Bar{\theta} = \frac{1}{n}\sum_\mu \theta_\mu$ to each party.
\end{enumerate}
\end{algorithm}

    \begin{figure}[ht]
    \centering
     \begin{subfigure}[b]{0.45\textwidth}
    \centering
    \includegraphics{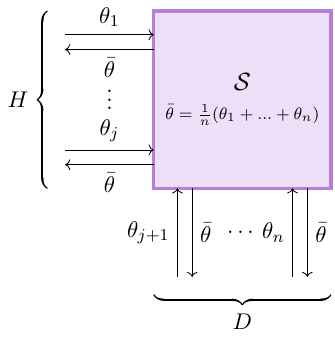}
             \caption{$\mathcal{S}$}
         \label{fig: ideal resource}
         \end{subfigure}
     \hfill
     \begin{subfigure}[b]{0.45\textwidth}
         \centering
    \includegraphics{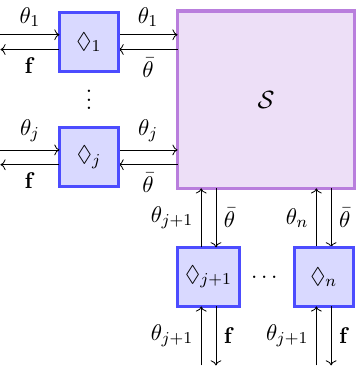}
             \caption{$\mathcal{S}_{\Diamond}$}
         \label{fig: filtered ideal resource}
    \end{subfigure}
    \caption{A representation of the ideal resource $\mathcal{S}$, without (a) and with (b) the filter $\Diamond$ applied, for $n$ parties, with left interfaces used for honest parties and interfaces below the resource for dishonest interfaces (the label is omitted on the right hand figure).}
    \end{figure}


This is shown in Fig.~\ref{fig: ideal resource}. The resource has $n$ interfaces, corresponding to the set $P$ of parties participating in the calculation, which we split into the honest set $H$ and the dishonest set $D$. We assume that the set $D$ may be collaborating, but that the members of $H$ do not. We could also consider an eavesdropper $E$, however, as the protocol does not allow for any participants to learn any information which is not made publicly available, this is not necessary -- that is, we do not need to consider whether certain information is kept secret from eavesdroppers, that is revealed to the participants. Thus, any adversarial behaviour can be modelled by the behaviour of $D$.

    \begin{figure}[ht]
    \centering
    \includegraphics{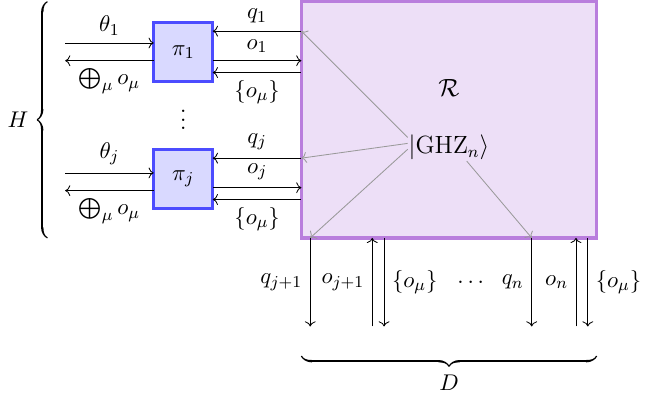}
    \caption{$\pi_H \mathcal{R}$: A representation of the concrete resource $\mathcal{R}$, where the converter $\pi_H$ corresponding to honest behaviour is applied only to the honest interfaces, $H$. We only show $N=1$, and the details of $\pi_H$ are in Alg.~\ref{honest protocol}. The label $q_\mu$ corresponds to a qubit, whereas $o_\mu$ are measurement outcomes.}
        \label{fig: concrete resource}
    \end{figure}

There are several possible variations on this resource, which will not make a significant impact on our security proof. For example, it is possible to add a mechanism by which any party can abort instead of inputting a parameter. It is also possible to slightly adjust the behaviour so that $\mathcal{S}$ receives the input $\theta_{\mu} \in [0, \pi/n)$ and returns the sum instead of the mean; as will become clear, this changes the accuracy of $\mathcal{R}$, but does not materially change the security proof.

For now, we will assume that we have access to a shared GHZ state across the network. In Section~\ref{sec: state verification}, we will consider the effects of relaxing this assumption. In this case, the concrete resource $\mathcal{R}$ has the same interfaces as $\mathcal{S}$, and can be described as in Alg.~\ref{concrete resource mean} (recall that $N$ is the number of rounds in the protocol). This could equally be considered as the composite of two resources, both the shared quantum state, and broadcasting classical bits, where we will later see how the first part can be replaced by a state verification protocol.
\begin{algorithm}
    \caption{Concrete resource for mean estimation, $\mathcal{R}$}\label{concrete resource mean}
    For $i = 1, ..., N$:
\begin{enumerate}
    \item Distribute one qubit of an $n$-qubit GHZ state to each of the interfaces.
    \item Receive from each interface the measurement outcome $o_{\mu, i}$.
    \item Return $\{ o_{\mu, i} \}_{\mu \in P}$ to every interface.
\end{enumerate}
\end{algorithm}

To interact with $\mathcal{R}$, we define the constructor $\pi$ as in Alg.~\ref{honest protocol}, which will act as both the honest protocol $\pi_H$ and the filter $\sharp_D$, modelling honest behaviour of the users. This has two interfaces, with the internal interface connecting to the resource $\mathcal{R}$, and the external interface open to the user to input their parameter. The protocol can be labelled by $\mu$ depending on the user with which it interacts. Together with $\mathcal{R}$, this is shown in Fig.~\ref{fig: concrete resource} (for $N=1$).
\begin{algorithm}
\caption{Honest protocol for mean estimation, $\mathcal{\pi}$ (for user $\mu = j$)}\label{honest protocol}
\begin{enumerate}
    \item Receive the parameter $\theta_j \in [0, 2\pi)$ from the user interface.
    \item For $i = 1, ..., N$:
\begin{enumerate}
    \item Receive a qubit from the resource interface.
    \item Implement the unitary $\hat{U}(\theta_j) = \ketbra{0}{0} + \mathrm{e}^{i \theta_j} \ketbra{1}{1}$ on the qubit.
    \item Measure the qubit in the $X$ basis, giving measurement result $o_{j, i}$.
    \item Return $o_{j, i}$ to the resource interface.
    \item Receive $\{ o_{\mu, i} \}_{\mu \in P}$ from the resource interface.
\end{enumerate}
\item Return the bitstring $\mathbf{p} \in \{ 0, 1\}^N$ to the user interface, where $p_i = \bigoplus_{\mu} o_{\mu, i}$.
\end{enumerate}
\end{algorithm}

The interfaces of both $\mathcal{S}$ and $\pi_H \mathcal{R} \pi_D$ both accept the same input (an angle), but $\mathcal{S}$ outputs a real number, whereas $\pi_H \mathcal{R} \pi_D$ outputs an $N$-bit string, and therefore a filter, $\Diamond$, is clearly required to compare these for the correctness proof. We define this in Alg.~\ref{filter mean}. The filtered resource (with $\Diamond$ applied to interfaces in $H$ and $D$) is shown in Fig.~\ref{fig: filtered ideal resource}.
\begin{algorithm}
    \caption{Filter for ideal resource, for mean estimation, $\Diamond$}\label{filter mean}
\begin{enumerate}
    \item Receive the angle $\theta_\mu$ from the user interface and return to the resource interface.
    \item Receive the angle $\Bar{\theta}$ from the resource interface.
    \item Return the bitstring $\mathbf{f} \in \{ 0, 1\}^N$ to the user interface, where $f_i$ is chosen randomly, with $\Pr(f_i = 0) = \frac{1}{2}\left(1 + \cos(n\Bar{\theta})\right)$.
\end{enumerate}
\end{algorithm}    

We now proceed to our first result, Theorem~\ref{thm: mean estimation}.
\begin{theorem} \label{thm: mean estimation}
Using the definitions of $\mathcal{S}, \mathcal{R}, \pi$ and $\Diamond$ from Algs.~\labelcref{ideal resource mean,concrete resource mean,honest protocol,filter mean}, $\pi_H$ constructs $\mathcal{S}_{\Diamond}$ from $\mathcal{R}_{\pi}$ exactly (that is, to within $\varepsilon = 0$).
\end{theorem}

\begin{proof}
\textbf{Correctness:} We aim to show that $\pi_H \mathcal{R} \pi_D$ and $\Diamond_H \mathcal{S} \Diamond_D$ cannot be distinguished (we are comparing Fig.~\ref{fig: concrete resource}, with the protocol $\pi$ applied to dishonest parties, to Fig.~\ref{fig: filtered ideal resource}). The open interfaces of $\Diamond_H \mathcal{S} \Diamond_D$ receive the $n$ parameters $\{ \theta_\mu \}$ and return the bitstring $\mathbf{f}$ as defined in Alg.~\ref{filter mean} (which is not fixed, but generated by a weighted coin flip). The interfaces of $\pi_H \mathcal{R} \pi_D$ receive the same input, and we must show that the output $\mathbf{p}$ is equal to $\mathbf{f}$. This is exactly the behaviour of parameter estimation~\cite{shettell2020graph, proctor2018multiparameter}. That is, the state generated at each of the $N$ repetitions of $\pi_H \mathcal{R} \pi_D$, prior to the $X$ basis measurement, is given by:
\begin{align}
    \frac{1}{\sqrt{2}} \bigotimes_\mu \hat{U} (\theta_\mu) \left(\ket{0}^{\otimes n} + \ket{1}^{\otimes n}\right) &= \frac{1}{\sqrt{2}} \left(\ket{0}^{\otimes n} + \mathrm{e}^{in\Bar{\theta}} \ket{1}^{\otimes n}\right) \\
    &= \frac{1}{2^{(n+1)/2}} \sum_{\ket{x} \in \{ \ket{+}, \ket{-} \}^{\otimes n}} (1 + (-1)^{\# \ket{-}} \mathrm{e}^{in\Bar{\theta}}) \ket{x},
\end{align}
where $\# \ket{-}$ is used to represent the number of qubits of $\ket{x}$ in the state $\ket{-}$. Hence, the parity of the $X$ basis measurements (which is stored as $p_i$) is random, with the probability of even parity given by $\frac{1}{4} \left| 1 + \mathrm{e}^{in\Bar{\theta}} \right|^2 = \frac{1}{2}\left(1 + \cos(n\Bar{\theta})\right)$ as required. More specifically, generating $\mathbf{f}$ should be a pseudorandom process so that both filters give the same output to the external interfaces, but otherwise the outputs ($\mathbf{p}$ and $\mathbf{f}$) cannot be distinguished.

\textbf{Security:} We use a simulator $\sigma$ on the open interfaces of $\mathcal{S}$ to reconstruct the information that would be available to the open interfaces of $\pi_H \mathcal{R}$, that is, $\pi_H \mathcal{R} \approx \sigma_H \mathcal{S} \sigma_D$ (so we need to modify Fig.~\ref{fig: ideal resource}, giving Fig.~\ref{fig: simulated ideal resource}, so that it matches Fig.~\ref{fig: concrete resource}). These simulators are defined in Algs.~\ref{honest simulator mean} and~\ref{dishonest simulator mean}.

\begin{algorithm}
    \caption{Simulator for honest participant $j$, for mean estimation, $\sigma_H$}\label{honest simulator mean}
\begin{enumerate}
    \item Receive the parameter $\theta_j$ from the external interface.
    \item Return $\theta_j$ to the internal interface.
    \item Receive $\Bar{\theta}_0$ from the internal interface (subscripts on $\bar{\theta}$ indicate the round number).
    \item For $i = 1,...,N$:
    \begin{enumerate}
        \item Return $0$ to the internal interface.
        \item Receive $\Bar{\theta}_i$ from the internal interface. 
        \item If $\Bar{\theta}_i = 0$, set $p_i = 0$, else set $p_i = 1$.
    \end{enumerate}
    \item Return the bitstring $\mathbf{p} \in \{ 0, 1\}^N$ to the external interface.
\end{enumerate}
\end{algorithm}

\begin{algorithm}
    \caption{Simulator for dishonest participants, for mean estimation, $\sigma_D$}\label{dishonest simulator mean}
\begin{enumerate}
    \item Return $0$ to the internal interface.
    \item Receive $\Bar{\theta}_0$ from the internal interface.
    \item For $i = 1,...,N$:
    \begin{enumerate}
        \item Randomly generate a length-$|H|$ bitstring $\mathbf{h}$, which has parity $h$.
        \item Distribute one qubit of the $|D|$-qubit state $\frac{1}{\sqrt{2}}(\ket{0}^{\otimes |D|} + (-1)^{h} e^{i n \Bar{\theta}_0} \ket{1}^{\otimes |D|})$ to each of the external interfaces.
        \item Receive from each external interface the bit $o_{\mu, i}$.
        \item Set $p_i = (\oplus_{\mu \in D} o_\mu \oplus h)$.
        \item Return $p_i \pi$ to the internal interface (for example, by returning $p_i \pi/|D|$ at each interface).
        \item Receive $\Bar{\theta}_i$ from the internal interface.
        \item Return $\{ o_{\mu} \}$ to the external interfaces, where $\{ o_\mu \}_{\mu \in H}$ are given by the bits of $\mathbf{h}$.
    \end{enumerate}
\end{enumerate}
\end{algorithm}

    \begin{figure}
    \centering
    \includegraphics{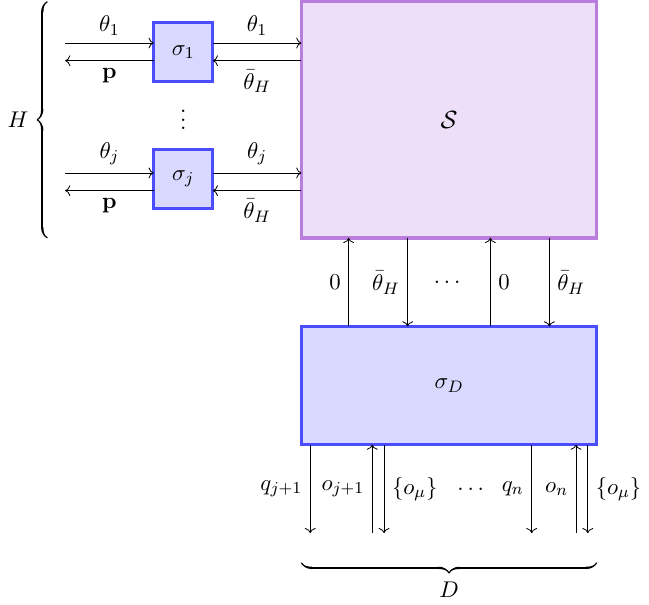}
    \caption{$\sigma_H \mathcal{S} \sigma_D$: A representation of the ideal resource $\mathcal{S}$, where the simulators $\sigma_H$ and $\sigma_D$ are applied to the corresponding interfaces. We only show $N=1$. For the security proof, this is compared to Fig.~\ref{fig: concrete resource}. We also omit the communication between simulators, which can be done through using $\mathcal{S}$, thus we only show the basic functionality of the simulators in reproducing the quantum state and measurement statistics.}
        \label{fig: simulated ideal resource}
    \end{figure}

Firstly, we consider $\pi_H \mathcal{R}$, the behaviour we wish to emulate. The external interfaces corresponding to honest parties receive an angle, $\theta_\mu$, and return the bitstring $\mathbf{p}$. The output of $\pi_H \mathcal{R}$ at the dishonest interfaces is firstly the state $\frac{1}{\sqrt{2}}(\ket{0}^{\otimes |D|} + (-1)^{h} \mathrm{e}^{i \sum_{\mu \in H} \theta_\mu} \ket{1}^{\otimes |D|})$. $h$ is the total parity of the measurement outcomes at the honest interfaces; these outcomes are perfectly random and uniformly distributed (as long as $|D|>0$). The dishonest interfaces of $\pi_H \mathcal{R}$ then receive the `measurement outcomes' $\{ o_\mu \}_{\mu \in D}$, which in fact can be any bits chosen by $D$. Finally, it returns the full set of measurement outcomes corresponding to both $H$ (generated internally by the interaction of $\pi$ with $\mathcal{R}$) and those received for $D$.

At the honest interfaces, the input to $\sigma_H \mathcal{S} \sigma_D$ is also an angle. At the dishonest interfaces, the simulator $\sigma_D$ must construct the state $\frac{1}{\sqrt{2}}(\ket{0}^{\otimes |D|} + (-1)^{h} \mathrm{e}^{i \sum_{\mu \in H} \theta_\mu} \ket{1}^{\otimes |D|})$. This is done by firstly generating a random bitstring corresponding to the `measurement outcomes' at $H$. The other required information in order to be able to construct this state is $\sum_{\mu \in H} \theta_\mu$, which is generated and sent to all parties by the first use of $\mathcal{S}$, the result of which is $\Bar{\theta}_0 = \frac{1}{n} \sum_{\mu \in H} \theta_\mu$.

The simulator $\sigma_H$ must return the bitstring $\mathbf{p}$ at the honest external interfaces, matching the total parities of the measurement outcomes at each round. This must be consistent at the external interfaces of both $\sigma_D$ and $\sigma_H$, as the distinguisher has access to all $n$ interfaces of $\sigma_H \mathcal{S} \sigma_D$. However, it is effectively decided by the bits input by $D$, as $h$ is decided in advance of distributing the state, and the measurement outcomes $\{ o_\mu \}_{\mu \in D}$ are received afterwards. If there is an inconsistency and either $p_i \neq \oplus_{\mu \in D} o_{\mu,i} \oplus h$, or if the state output is $\frac{1}{\sqrt{2}}(\ket{0}^{\otimes |D|} + (-1)^{1 \oplus h} e^{i \sum_{\mu \in H} \theta_\mu} \ket{1}^{\otimes |D|})$, the distinguisher would be able to use this to distinguish between $\pi_H \mathcal{R}$ and $\sigma_H \mathcal{S} \sigma_D$. Therefore, it is required to use the resource $\mathcal{S}$ to signal the bit $p_i$ to $\sigma_H$ for each round (this can be done in multiple ways, but an example method is used in the suggested construction in Alg.~\ref{dishonest simulator mean}). Therefore, these simulators can render the ideal resource indistinguishable from the concrete resource, $\pi_H \mathcal{R} \approx \sigma_H \mathcal{S} \sigma_D$.

It is worth noting that there is no requirement for the parities of $\mathbf{p}$ to match $\mathbf{f}$ (that is, to genuinely implement a mean calculation). In fact, this is the case if the parties of $D$ follow the honest protocol, but dishonest behaviour by the parties has the effect of inputting incorrect values for their local parameters, which harms the integrity of the protocol but not the privacy. 
\end{proof}

In this construction, we use $\mathcal{S}$ to send messages between the simulators. In total, the resource $\mathcal{S}$ is called $N + 1$ times by the simulators. If the resource $\mathcal{R}$ (and associated constructors) were defined so that all $N$ of the states are distributed, then all measurement outcomes collected before being sent across the network, then it would be possible to only use $\mathcal{S}$ twice, although this would require a more complicated scheme for encoding and sending the measurement outcomes. An alternative method is used in~\cite{colisson2024all}, in which ideal resources are equipped with a communication channel that can forward any message, which is used to mimic non-local resources -- it is straightforward to see how this could be used to replace the multiple uses of $\mathcal{S}$ in our proof. Similarly, this could be use to decompose the monolithic simulator $\sigma_D$ into local simulators for each party; the abstract cryptography formalism allows for monolithic simulators, although the universal composability framework was originally intended to only use local simulators for each dishonest user. 


The correctness aspect of this proof indicates that, using a GHZ state, we can exactly implement a scheme described by $\Diamond_H \mathcal{S} \Diamond_D$ (that is, estimating the mean using the bitstring $\mathbf{f}$), whereas the security proof indicates that this scheme is at least as secure as $\mathcal{S}$. This aligns with the definition of privacy introduced in other works~\cite{shettell2022private, bugalho2025private, hassani2025privacy}, that participants (as individuals or a collaborating group) in the scheme cannot learn any information that cannot be gained from knowledge of their own parameter(s) and the desired function alone. Using this scheme, $\Bar{\theta}$ can be estimated with a variance of $1/n^2 N$, however we note that there are $2n$ possible values (as we estimate mod $\pi$). If $\theta/2n$ is instead given as the input and the sum is returned, there is one possible outcome, with variance $4/N$.

An alternative definition of $\mathcal{S}$ could instead be used that more closely resembles $\mathcal{R}$ and therefore does not require filters for the correctness proof. The security of this relies on limited knowledge of $\Bar{\theta}$ being accessible, and does not have perfect information-theoretic security, however this can be guaranteed by a further modification of one participant hiding their measurement outcome, or encrypting it. Nonetheless, we will continue to focus on the simpler definition of $\mathcal{S}$ that more closely resembles the existing privacy definition in the next section, where we generalise this proof to other functions.

\section{General function estimation} \label{sec: general fn}

Now, instead of aiming to calculate the mean of the parameters, we will consider a more general setting, where we wish to estimate $f(\boldsymbol{\theta}) = \mathbf{a}\cdot\boldsymbol{\theta}$. We will assume that all elements of $\mathbf{a}$ are non-zero (otherwise the associated user can be removed from the protocol). We will also assume, w.l.o.g., that $\mathbf{a}$ is a unit vector. Once again we will begin by presenting the required components of our construction, in Algs.~\labelcref{ideal resource general,concrete resource general,honest protocol general}, although the filter $\Diamond$ is omitted for now. We will use the same notation, e.g. in this section, $\mathcal{S}$ should be understood to mean the ideal resource for general function estimation, and not for the estimation of the mean.

\begin{algorithm}
    \caption{Ideal resource for general function estimation, $\mathcal{S}$}\label{ideal resource general}
\begin{enumerate}
    \item Receive the parameter $\theta_\mu \in [0, 2\pi)$ from each party.
    \item Return $f(\boldsymbol{\theta}) = \mathbf{a}\cdot\boldsymbol{\theta}$ to each party.
\end{enumerate}
\end{algorithm}

\begin{algorithm}
    \caption{Concrete resource for general function estimation, $\mathcal{R}$}\label{concrete resource general}
    For $i = 1, ..., N$:
\begin{enumerate}
    \item Distribute one qubit of an $n$-qubit state $\rho$ to each of the interfaces.
    \item Receive from each interface the measurement outcome $o_{\mu, i}$.
    \item Return $\{ o_{\mu, i} \}_{\mu \in P}$ to every interface.
\end{enumerate}
\end{algorithm}

\begin{algorithm}
\caption{Honest protocol for general function estimation, $\mathcal{\pi}$ (for user $\mu = j$)}\label{honest protocol general}
\begin{enumerate}
    \item For $i = 1, ..., N$:
\begin{enumerate}
    \item Receive the parameter $\theta_{j} \in [0, 2\pi)$ from the external interface.
    \item Receive a qubit from the internal interface.
    \item Implement the quantum channel $\Lambda_j(\theta_j)$ on the qubit.
    \item Measure the qubit in the computational basis, giving measurement result $o_{j}$.
    \item Return $o_{j}$ to the internal interface.
    \item Receive $\{ o_{\mu, i} \}_{\mu \in P}$ from the internal interface.
\end{enumerate}
\item Return the bitstring $\mathbf{p} \in \{ 0, 1\}^N$ to the external interface, where $p_i = g(\{o_{\mu, i}\}_{\mu})$.
\end{enumerate}
\end{algorithm}

The concrete resource $\mathcal{R}$ and honest protocol $\pi$ define the dynamics of the protocol through the initial state $\rho$ (which we assume is an $n$-qubit state, but which can be generalised to consider a larger multimode state) and the quantum channel $\Lambda_\mu(\theta_\mu)$ (we assume w.l.o.g. that the local measurements are carried out in the computational basis, although again this can be generalised to other measurements, as long as they are independent of $\{ \theta_\mu \}$, and we also assume w.l.o.g. that $\Lambda(0) = \mathds{1}$). In general these channels can be specific to each user (hence $\Lambda$ is labelled by $\mu$), but they should be known to everyone in the network. We also do not specify the function $g$ which is used to combine measurement outcomes (which was previously addition modulo 2). However, the function $g(\{ o_\mu \})$ should require the measurement bit from every single user -- otherwise, the parameters of some users would be irrelevant to the calculation. These dynamics can be chosen according to either~\cite{bugalho2025private} or~\cite{hassani2025privacy}, with the former considering an appropriate set of states for use in this protocol, and we will consider the security proofs of these schemes separately. 

Let us consider the requirements on these dynamics that must be fulfilled for the constructive security proof, $\pi_H \mathcal{R} \approx_{\varepsilon} \sigma_H \mathcal{S} \sigma_D$. Once again, the output of $\pi_H \mathcal{R}$ at the honest interfaces is generally straightforward to reproduce, and therefore $\sigma_H$ can be constructed in a similar way to the GHZ case. Hence let us consider the simulator $\sigma_D$.

As in Section~\ref{sec: mean estimation}, the simulator needs to produce both the bitstring over the honest parties, $\mathbf{h}$, and the partial state over the dishonest parties (up to normalisation):
\begin{equation} \label{eq: dishonest output state}
\rho_{\text{out}, D} := \Tr_H\left[\bigotimes_{\mu \in H} \Pi_{o_\mu}\Lambda_\mu(\theta_\mu) \rho \left(\bigotimes_{\nu \in H} \Pi_{o_\nu} \Lambda_\nu(\theta_\nu) \right)^\dagger\right],
\end{equation}
in which $\Pi_{o_\mu}$ is the projector associated with measurement outcome $o_\mu$. Once again we will use the ideal resource $\mathcal{S}$ to communicate the function $f(\theta_H)$ (that is, the function of $\theta$ where every $\{\theta_\mu\}_{\mu \in D}$ is set to 0) to $\sigma_D$, and therefore we can assume that $\sigma_D$ has access to this information about the parameters of honest parties, but no more.

Given that we do not specify which parties are honest or dishonest, let us first consider how to construct the global state $\rho(\boldsymbol{\theta}) := \bigotimes_{\mu \in P} \Lambda_\mu(\theta_\mu) \rho \left(\bigotimes_{\nu \in P} \Lambda_\nu(\theta_\nu) \right)^\dagger$. This is one of a set of states, parametrised by $\boldsymbol{\theta}$, which we wish to construct only using knowledge of $f(\boldsymbol{\theta})$. Thus we can always construct a state $\rho(\boldsymbol{\theta}')$, where $f(\boldsymbol{\theta}') = f(\boldsymbol{\theta})$, but individual parameters may otherwise differ -- that is, $\boldsymbol{\theta}' - \boldsymbol{\theta} \propto \mathbf{b}$, where $\mathbf{a} \cdot \mathbf{b} = 0$ and $\mathbf{b}$ is a unit vector. 

For a single repetition of the protocol ($N = 1$), we will show that the distinguishing advantage is upper bounded by $\max_{\boldsymbol{\theta}, \boldsymbol{\theta}'}(T(\rho(\boldsymbol{\theta}), \rho(\boldsymbol{\theta}')))$ where $T(\rho, \rho')$ is the trace distance, $\frac{1}{2} \norm{\rho - \rho'}_1$. Over multiple repetitions, this becomes $\max_{\boldsymbol{\theta}, \boldsymbol{\theta}'}(T(\rho(\boldsymbol{\theta})^{\otimes N}, \rho(\boldsymbol{\theta}')^{\otimes N}))$~\cite{helstrom1969quantum}, which we will consider in more detail later on.

For now let us restrict ourselves to the case $N=1$, for which we have the following result, which we can use to evaluate the security of certain dynamics:
\begin{lemma} \label{lemma: security}
    Consider a parameter estimation scheme using the definitions of $\mathcal{R}$, $\mathcal{S}$, and $\pi$ from Algs.~\labelcref{concrete resource general,ideal resource general,honest protocol general} -- that is, using encoding dynamics $\Lambda_\mu(\theta_\mu)$ on state $\rho$ to produce state $\rho(\boldsymbol{\theta})$, used to estimate $f(\boldsymbol{\theta}) = \mathbf{a} \cdot \boldsymbol{\theta}$. In the case that $N = 1$, it is possible to define simulators with security $\varepsilon$, as defined in Def.~\ref{defn: constructive cryptography}:
    \begin{equation} \label{eq: lemma}
       \max_{\boldsymbol{\theta}, \boldsymbol{\theta}'}\left(T(\rho(\boldsymbol{\theta}), \rho(\boldsymbol{\theta}'))\right) \leq \varepsilon  \implies \exists \; \sigma_D, \sigma_H \;\text{s.t.} \; \pi_H \mathcal{R} \approx_\varepsilon \sigma_H \mathcal{S} \sigma_D,
    \end{equation}
    where $\boldsymbol{\theta}, \boldsymbol{\theta}' \in \{ 2 \pi \}^n$ are any $n$-element vectors of angles such that $\mathbf{a} \cdot (\boldsymbol{\theta}' - \boldsymbol{\theta}) = 0$.
\end{lemma}

\begin{proof}
    The simulator $\sigma_H$ is defined as in Alg.~\ref{honest simulator mean}, but with the modification that every instance of $\Bar{\theta}$ should instead be replaced by $f(\boldsymbol{\theta})$ (and with $N = 1$). The simulator $\sigma_D$ is a modified version of Alg.~\ref{dishonest simulator mean}, described in Alg.~\ref{dishonest simulator general}.

    \begin{algorithm}
    \caption{Simulator for dishonest participants, for general function estimation, $\sigma_D$}\label{dishonest simulator general}
\begin{enumerate}
    \item Return $0$ to the internal interface.
    \item Receive $f_0(\boldsymbol{\theta})$ from the internal interface.
        \item Using $f_0(\boldsymbol{\theta})$, select angles $\theta_H = \{ \theta_\mu \}_{\mu \in H}$ such that $f(\boldsymbol{\theta}_H) = f_0(\boldsymbol{\theta})$.
        \item Construct the state $\boldsymbol{\Lambda}_H (\boldsymbol{\theta}_H) \rho \boldsymbol{\Lambda}_H^\dagger (\boldsymbol{\theta}_H)$ (where $\boldsymbol{\Lambda}_H (\boldsymbol{\theta}_H)= \bigotimes_{\mu \in H} \Lambda_\mu(\theta_\mu)$). 
        \item Measure the parties associated with $H$ in the computational basis, giving measurement outcomes $\{ o_\mu\}_{\mu \in H}$.
        \item Distribute one qubit of the resulting $|D|$-qubit state to each of the external interfaces.
        \item Receive from each external interface the bit $o_{\mu}$.
        \item Set $p = g(\{o_{\mu}\})$.
        \item Return $p \pi$ to the internal interface (for example, by returning $p \pi/|D|$ at each interface).
        \item Receive $f(\boldsymbol{\theta})$ from the internal interface.
        \item Return $\{ o_{\mu} \}$ to the external interfaces.
\end{enumerate}
\end{algorithm}

As is the case in the proof of Theorem~\ref{thm: mean estimation}, the external honest interfaces of $\pi_H \mathcal{R}$ and $\sigma_H \mathcal{S} \sigma_D$ receive an angle, and output the bit $\mathbf{p}$. As this depends on the outcomes generated and received by $\sigma_D$, we once again use $\mathcal{S}$ multiple times to communicate between the simulators.

We must compare the states output to the dishonest interfaces of $\pi_H \mathcal{R}$ and $\sigma_H \mathcal{S} \sigma_D$, as well as the measurement outcomes $\{ o_\mu \}_{\mu \in H}$, which differ only by the choice of parameters of $H$ (where we will label the choice input by the distinguisher to $\pi_H$ as $\boldsymbol{\theta}_H$, and the parameters chosen by $\sigma_D$ as $\boldsymbol{\theta}'_H$). Let $\rho(\boldsymbol{\theta}_H)$ be the encoded state $\boldsymbol{\Lambda}_H (\boldsymbol{\theta}_H) \rho \boldsymbol{\Lambda}_H^\dagger (\boldsymbol{\theta}_H)$. Then, by the Holevo-Helstrom theorem~\cite{helstrom1969quantum, holevo1973statistical}, the distinguishing advantage is upper bounded by the maximum of $T(\rho (\boldsymbol{\theta}_H), \rho (\boldsymbol{\theta}'_H))$. However:
\begin{equation}
    \max\left(T(\rho (\boldsymbol{\theta}_H),\rho (\boldsymbol{\theta}'_H))\right) \leq \max_{\boldsymbol{\theta}, \boldsymbol{\theta}'}\left(T(\rho(\boldsymbol{\theta})), \rho(\boldsymbol{\theta}')\right) \leq \varepsilon,
\end{equation}
as $\rho (\boldsymbol{\theta}_H)$ and $\rho (\boldsymbol{\theta}'_H)$ are in fact special cases of the states $\rho(\boldsymbol{\theta})$ and $\rho(\boldsymbol{\theta}')$ which we consider, where every $\theta_{\mu}$ for $\mu \in D$ is set to 0. Hence, the maximum distinguishing advantage is $\varepsilon$.

\end{proof}

\begin{remark}
Although we have presented this within the context of networked parameter estimation, Lemma~\ref{lemma: security} can be applied more generally. Consider a scheme where we use the state $\rho(\boldsymbol{\theta})$, which encodes information $\boldsymbol{\theta} \in \Phi$, but where $\Phi$ is the condition which is intended to be revealed by the scheme. For example, in the parameter estimation protocol, $\Phi$ is the set of parameter assignments with the same value for $\boldsymbol{\theta} \cdot \mathbf{a}$. If Eq.~\ref{eq: lemma} holds for any $\boldsymbol{\theta}' \in \Phi$, then we can use our knowledge of $\Phi$ to construct $\rho(\boldsymbol{\theta}')$, and hence the proof follows along similar lines to the proof of Lemma~\ref{lemma: security} -- that is, the information revealed is $\Phi$, and not $\boldsymbol{\theta}$.

For example, consider the definition of anonymity given in~\cite{grasselli2022secure}, which was applied to parameter estimation in~\cite{dejong2025anonymous}. Then we require that states should be indistinguishable under any change of identities across the network. Hence, we only need to show that Eq.~\ref{eq: lemma} is fulfilled for any assignment of user identities, and then we can construct the resources and converters in a similar way, without revealing information about user identities.
\end{remark}

In order to consider the distinguishing probability for implementations of the protocol with $N > 1$, we can still use Lemma~\ref{lemma: security}, but this time bearing in mind that this is across $N$ copies of $\rho(\boldsymbol{\theta})$, and that the distinguisher has the ability to make entangling measurements. We find that:
\begin{equation}
\begin{split}
    \max_{\boldsymbol{\theta}, \boldsymbol{\theta}'}\left(T\left(\rho(\boldsymbol{\theta}), \rho(\boldsymbol{\theta}')\right)\right) &\leq \varepsilon  \\ \implies \max_{\boldsymbol{\theta}, \boldsymbol{\theta}'}\left(T\left(\rho(\boldsymbol{\theta})^{\otimes N}, \rho(\boldsymbol{\theta}')^{\otimes N}\right)\right) &\leq \sqrt{1 - F\left(\rho(\boldsymbol{\theta})^{\otimes N}, \rho(\boldsymbol{\theta}')^{\otimes N}\right)}\\ 
    &= \sqrt{1 - F(\rho(\boldsymbol{\theta}), \rho(\boldsymbol{\theta}'))^N} \\
    & \leq \sqrt{1 - (1 - \varepsilon^2)^N}. 
    \end{split}
\end{equation}

At this stage, we must consider the dynamics of the parameter estimation protocol in more detail in order to be able to judge their security. We do this using the formalisms of two prior works: in Sections~\ref{sec: luis defn} and~\ref{sec: majid defn}, we consider the definition of privacy in~\cite{bugalho2025private} and~\cite{hassani2025privacy} respectively. These both use the quantum Fisher information, which is introduced in Appendix Section~\ref{sec: QFI appendix}. Further work \cite{bugalho2025private} focuses on states which can be used as the basic resource of the scheme while still guaranteeing privacy, whereas~\cite{hassani2025privacy} focuses on the encoding dynamics and the effect of error.

\subsection{Quasi-privacy definition: Bugalho et al.} \label{sec: luis defn}

We will now consider the parameter estimation schemes described in~\cite{bugalho2025private}. This work introduces a metric for the privacy of certain implementations, given in terms of the quantum Fisher information matrix (QFIm)~\cite{sidhu2020geometric, liu2020quantum}, $\mathcal{Q}$, which is defined for a particular state (i.e. $\rho(\boldsymbol{\theta})$), given with respect to a set of parameters (i.e. $\{ \theta_\mu \}$), and is a measure of the available information about a particular parameter that can be extracted from a state by the appropriate measurements. 

They introduce the following privacy measure:
\begin{definition}[Bugalho et al., \cite{bugalho2025private}]
    The privacy measure of a multi-parameter estimation problem, which results in a quantum Fisher information matrix $\mathcal{Q}$, with respect to a target linear function $f(\boldsymbol{\theta}) = \mathbf{a} \cdot \boldsymbol{\theta}$ with $\norm{\mathbf{a}} = 1$ is given by:
    \begin{equation}
        \mathcal{P} (\mathcal{Q}, \mathbf{a}) = \frac{\mathbf{a}^\text{T} \mathcal{Q} \mathbf{a}}{\Tr(\mathcal{Q})} \equiv \frac{\Tr( \mathcal{Q} \mathbf{a} \mathbf{a}^\text{T})}{\Tr(\mathcal{Q})} = \frac{\Tr(\mathcal{Q} W_{\textbf{a}})}{\Tr(\mathcal{Q})}.
    \end{equation}
\end{definition}

In particular, $\mathcal{P} = 1$ means complete privacy -- that is, a state that does not provide dishonest participants with any information that is not calculable from $f(\boldsymbol{\theta})$ and their own parameters. $\mathcal{P}$ varies between 0 and 1.


This definition is composably secure, which we show in a similar way to our previous construction, although note that we do not explicitly define the filter $\Diamond$. However, the distinguishing advantage $\varepsilon$ scales as $\sqrt{1 - \mathcal{P}^2}$.

\begin{theorem} \label{thm: luis defn composably private}
Using the definitions of $\mathcal{S}, \mathcal{R}$ and $\pi$ from Algs.~\labelcref{ideal resource general,concrete resource general,honest protocol general} for $N=1$, there exists some $\Diamond$ such that $\pi_H$ constructs $\mathcal{S}_{\Diamond}$ from $\mathcal{R}_{\pi}$ to within $\sqrt{1 - \mathcal{P}^2(\mathcal{Q}, \mathbf{a})}$, where $\mathcal{Q}$ is the QFIm of the encoded state $\rho(\boldsymbol{\theta})$ with respect to $\boldsymbol{\theta}$.    
\end{theorem}

\begin{proof}
    \textbf{Security:} We first aim to show that $\pi_H \mathcal{R} \approx_{\varepsilon} \sigma_H \mathcal{S} \sigma_D$ for some $\varepsilon$. Using Lemma~\ref{lemma: security}, this is achieved by showing that $\max_{\boldsymbol{\theta}, \boldsymbol{\theta}'}\left(T(\rho(\boldsymbol{\theta}), \rho(\boldsymbol{\theta}'))\right) \leq \varepsilon$, where we will find that $\varepsilon = \sqrt{1 - \mathcal{P}^2(\mathcal{Q}, \mathbf{a})}$. Recall that $\boldsymbol{\theta}, \boldsymbol{\theta}'$ are any $n$-element vectors of angles such that $\mathbf{a} \cdot (\boldsymbol{\theta}' - \boldsymbol{\theta}) = 0$, and we will say that $\boldsymbol{\theta}' - \boldsymbol{\theta} \propto \mathbf{b}$.

    We will consider the Bures distance (see Appendix Section~\ref{sec: QFI appendix}), $D_B (\rho, \sigma)$, a measure of the distance between two states which ranges between 0 and $\sqrt{2}$. For two states $\rho$ and $\sigma$~\cite{paris2009quantum}:
    \begin{equation} \label{eq: fidelity bures}
        F(\rho, \sigma) = \left( 1 - \frac{1}{2} D^2_B (\rho, \sigma)\right)^2,
    \end{equation}
    where $F(\rho, \sigma)$ is the fidelity. For a state $\rho(\boldsymbol{\theta})$, the Bures metric -- the Bures distance between states separated by an infinitesimally small change in parameters -- can be defined in relation to the QFIm with regards to the parameters $\boldsymbol{\theta}$ using:
    \begin{equation}
        D^2_B (\rho(\boldsymbol{\theta}), \rho(\boldsymbol{\theta} + d\boldsymbol{\theta})) = \frac{1}{4} \mathcal{Q}_{\mu \nu} (\rho(\boldsymbol{\theta})) d\theta_\mu d\theta_\nu.
    \end{equation}
    Thus, we have a relation between the QFIm and the trace distance (which can be bounded using the fidelity).

    We can always construct an orthogonal matrix $A$, which has the vector $\mathbf{a}$ as the first column, $\mathbf{b}$ as the second column, and where the further columns are given by the other orthonormal vectors that form a basis. This permits the change of basis of the QFIm:
    \begin{equation}
        \Tilde{\mathcal{Q}} = A^\text{T} \mathcal{Q} A,
    \end{equation}
    where $\Tilde{\mathcal{Q}}$ is now expressed in terms of a different parameter set $\Tilde{\boldsymbol{\theta}}$, where $\Tilde{\theta}_1 = \mathbf{a}\cdot \boldsymbol{\theta}$ and $\Tilde{\theta}_2 = \mathbf{b}\cdot \boldsymbol{\theta}$, etc. Therefore, we can see that, by construction, $\boldsymbol{\theta}' - \boldsymbol{\theta} \propto \Tilde{\theta}_2$. Hence:
    \begin{equation} \label{eq: Bures metric}
        D^2_B (\rho(\boldsymbol{\theta}), \rho(\boldsymbol{\theta} + d\Tilde{\theta}_2)) = \frac{1}{4} \Tilde{\mathcal{Q}}_{2 2} d^2\Tilde{\theta}_2.
    \end{equation}

    In order to use Eq.~\ref{eq: Bures metric} to find the Bures distance, the Bures metric must be integrated along the geodesic between $\rho(\boldsymbol{\theta})$ and $\rho(\boldsymbol{\theta}')$. Assuming that parameters are encoded by local unitaries (or indeed, by local CPTP maps which can be applied cyclically over the parameter space), the QFI is constant, regardless of the parametrisation of $\rho(\boldsymbol{\theta})$~\cite{liu2020quantum}. Hence, $\tilde{Q}_{22}$ is also constant over the parameter space. Thus, integrating over a parameter change in the direction of $\mathbf{b}$, the Bures distance must be proportional to $\sqrt{\Tilde{\mathcal{Q}}_{2 2}}$.

    Note that the choice of parametrisation is not unique, and therefore this relationship must be appropriately normalised so that the Bures distance is between 0 and $\sqrt{2}$. The highest possible value of $\Tilde{\mathcal{Q}}_{22}$ is $\Tr(\mathcal{Q})$, which is constant under any reparametrisation, and represents the case where the information available aligns only with the function $\mathbf{b} \cdot \boldsymbol{\theta}$. Hence:
    \begin{equation}
        D^2_B (\rho(\boldsymbol{\theta}), \rho(\boldsymbol{\theta} + d\Tilde{\theta}_2)) = \frac{2 \Tilde{\mathcal{Q}}_{2 2}}{\Tr(\mathcal{Q})}.
    \end{equation}

    Furthermore, it is straightforward to show that given that $\mathcal{Q}$ is real, symmetric, and positive semidefinite matrix, $\mathbf{b}^{\text{T}} \mathcal{Q} \mathbf{b} = \Tilde{\mathcal{Q}}_{22} \leq \Tr(\mathcal{Q}) - \mathbf{a}^{\text{T}} \mathcal{Q} \mathbf{a}$, and thus:
    \begin{equation}
        D^2_B (\rho(\boldsymbol{\theta}), \rho(\boldsymbol{\theta} + d\Tilde{\theta}_2)) \leq 2 - \frac{2 \mathbf{a}^{\text{T}} \mathcal{Q} \mathbf{a}}{\Tr(\mathcal{Q})} = 2 \left( 1 - \mathcal{P}(\mathcal{Q}, \mathbf{a})\right).
    \end{equation}

    By using $T(\rho, \sigma) \leq \sqrt{1 - F(\rho, \sigma)}$ and Eq.~\ref{eq: fidelity bures}, we have:
        \begin{equation} \label{eq: equiv of qfi and trace distance}
        \max_{\boldsymbol{\theta}, \boldsymbol{\theta}'}\left(T(\rho(\boldsymbol{\theta}), \rho(\boldsymbol{\theta}'))\right) \leq \sqrt{1 - \mathcal{P}^2(\mathcal{Q}, \mathbf{a})}.
    \end{equation} 
    Therefore, using Lemma~\ref{lemma: security}, $\pi_H \mathcal{R} \approx_{\sqrt{1 - \mathcal{P}^2}} \sigma_H \mathcal{S} \sigma_D$.
    
    \textbf{Correctness:} The aim is to construct an appropriate filter $\Diamond$ such that $\pi_H \mathcal{R} \pi_D \approx_{\varepsilon} \Diamond_H \mathcal{S} \Diamond_D$. This is possible if, after receiving $f(\boldsymbol{\theta})$ from the internal (resource) interface, the filter can produce the output $\mathbf{f}$ (or the single bit $f$ for $N=1$) which has the same distribution as the bit $p$ which is output at the external interfaces of $\Diamond_H \mathcal{S} \Diamond_D$. A simple way that this would be possible would be, for example, to construct the state $\rho(\boldsymbol{\theta})$ and then perform the appropriate measurements. As we have seen by Eq.~\ref{eq: equiv of qfi and trace distance}, it is possible, using $f(\boldsymbol{\theta})$, to construct a state which is within trace distance $\sqrt{1 - \mathcal{P}^2(\mathcal{Q}, \mathbf{a})}$ of $\rho(\boldsymbol{\theta})$, and hence this is the maximum distinguishing advantage possible when comparing the results $f$ and $p$. (Once again $\mathcal{S}$ could be used to communicate between filters and ensure that the output is the same for both, however with well-defined dynamics it would not be necessary to construct the state $\rho(\boldsymbol{\theta})$ as the outcome probabilities could be calculated directly, as is the case in Section~\ref{sec: mean estimation}, and thus a pseudo-random process could be used to generate identical outcomes for both simulators.)
\end{proof}

As we do not define a particular $\Diamond$, this proof does not make any claims on the efficiency, or accuracy, of the quantum estimation scheme, as $\mathcal{S}_{\Diamond}$ represents what can actually be achieved by the realistic protocol. Indeed, for an encoding (such as using the state $\ket{0}^{\otimes n}$, but with the rotations and measurements of Section~\ref{sec: mean estimation}) which has no phase sensitivity, an appropriate $\Diamond$ could be constructed to satisfy the composable privacy definition of Thm.~\ref{thm: luis defn composably private}, but this would be useless for distributed sensing. On the other hand, metrological quantum advantage in networked sensing is explored in works such as~\cite{proctor2017networked}, and any states presented in these works can then be analysed for their privacy using the framework given here.

\subsection{Quasi-privacy definition: Hassani et al.} \label{sec: majid defn}

We now consider the definition of privacy, and quasi-privacy, presented in~\cite{hassani2025privacy}.

Firstly, complete privacy is also defined in this work in terms of the QFI matrix, where this is achieved if:
\begin{equation} \label{eq: hassani privacy condition}
    \mathcal{Q}_{\mu \nu} \propto a_{\mu} a_{\nu}.
\end{equation}
Returning to the Bures metric, and using the previous argument that we change $\boldsymbol{\theta}$ in the direction $\mathbf{b}$:
\begin{equation}
    D^2_B (\rho(\boldsymbol{\theta}), \rho(\boldsymbol{\theta} + d\boldsymbol{\theta})) = \frac{1}{4} a_{\mu} a_{\nu} b_{\mu} b_{\nu} (d\theta)^2,
\end{equation}
however, as $\mathbf{a}\cdot \mathbf{b} = 0$, this means that $D^2_B (\rho(\boldsymbol{\theta}), \rho(\boldsymbol{\theta} + d\boldsymbol{\theta})) = 0$. Therefore, using the same arguments as the previous section, it is straightforward to see that this definition of privacy is composable.

Imperfect privacy is considered from several perspectives in the paper, in particular focusing on the example of using GHZ states for estimating the mean. 

\begin{definition}[Hassani et al., \cite{hassani2025privacy}] \label{defn: hassani privacy}
    Consider a multi-parameter mean estimation problem, which uses initial state $\rho$, giving the state $\rho(\boldsymbol{\theta})$ after encoding. The $\varepsilon$-privacy may be quantified:
\begin{equation}
\varepsilon = \norm{\partial_\mu \rho(\boldsymbol{\theta}) - \partial_\nu \rho(\boldsymbol{\theta})}_1
\end{equation}
    where $ \partial_\mu = \partial / \partial \theta_\mu$.
\end{definition}

One consequence of this definition is that if local parameters are encoded using the unitary evolutions $\hat{U} (\theta_\mu) = \exp(-i \hat{H}_\mu (\theta_\mu))$, where $\hat{H}_\mu (\theta_\mu)$ is a Hermitian operator that acts non-trivially on the Hilbert space of each (local) quantum sensor, then:
    \begin{equation} \label{eq: hamiltonian}
        \varepsilon = \norm{[\hat{H}'_\mu - \hat{H}'_\nu, \rho]}_1, 
    \end{equation}
where $\hat{H}_\mu' = \partial_\mu \hat{H}_\mu = \partial \hat{H}_\mu / \partial \theta_\mu$.

The definition comes from quantifying the requirement that:
\begin{equation}
    \norm{\partial_\mu \rho(\boldsymbol{\theta}) - \partial_\nu \rho(\boldsymbol{\theta})}_1 \propto |a_\mu - a_\nu | \; \forall \; \mu, \nu,
\end{equation}
which is itself derived from the condition in Eq.~\ref{eq: hassani privacy condition}. Hence, we will consider this more general case, although the details of the derivation of Defn.~\ref{defn: hassani privacy} can be found in~\cite{hassani2025privacy}. Once again, we will consider only the $N=1$ case at this stage. This proof closely follows Thm.~\ref{thm: luis defn composably private}, and hence many details are omitted.

\begin{theorem} \label{thm: hassani defn composably private}
Using the definitions of $\mathcal{S}, \mathcal{R}$ and $\pi$ from Algs.~\labelcref{ideal resource general,concrete resource general,honest protocol general} for $N=1$, if the QFIm of the encoded state $\rho(\boldsymbol{\theta})$ with respect to $\boldsymbol{\theta}$ satisfies:
\begin{equation} \label{eq: hassani proportionality}
    |\mathcal{Q}_{\mu \nu} - k a_\mu a_\nu | \leq \epsilon \; \forall \; \mu, \nu
\end{equation}
for some constant $k$, then there exists some $\Diamond$ such that $\pi_H$ constructs $\mathcal{S}_{\Diamond}$ from $\mathcal{R}_{\pi}$ to within $n\epsilon/\Tr(\mathcal{Q})$.    
\end{theorem}

\begin{proof}
\textbf{Security:} Let us say that:
\begin{equation}
    \mathcal{Q}_{\mu \nu} = ka_\mu a_\nu + \epsilon_{\mu \nu},
\end{equation}
where $|\epsilon_{\mu \nu}| \leq \epsilon$. As in the proof of Thm.~\ref{thm: luis defn composably private}, we will consider the QFIm element that corresponds to a change of parameters in the direction of $\mathbf{b}$: $\Tilde{\mathcal{Q}}_{22} = \mathbf{b}^{\text{T}} \mathcal{Q} \mathbf{b}$. We can upper bound this:
\begin{equation}
    \begin{split}
        \Tilde{\mathcal{Q}}_{22} &= \sum_{\mu, \nu} b_\mu \mathcal{Q}_{\mu \nu} b_\nu \\
        &= \sum_{\mu, \nu} b_\mu (k a_\mu a_\nu + \epsilon_{\mu \nu}) b_\nu \\
        &\leq \epsilon \sum_\mu |b_\mu| \sum_\nu |b_\nu| \\
        &\leq n \epsilon.
    \end{split}
\end{equation}

As before, we use the Bures distance:
\begin{equation}
        D^2_B (\rho(\boldsymbol{\theta}), \rho(\boldsymbol{\theta} + d\Tilde{\theta}_2)) = \frac{2 \Tilde{\mathcal{Q}}_{2 2}}{\Tr(\mathcal{Q})} \leq \frac{2 n \epsilon}{\Tr(\mathcal{Q})}.
    \end{equation}
This is normalised to be between 0 and 2 as $\Tr(\mathcal{Q}) \geq \mathbf{b}^{\text{T}} \mathcal{Q} \mathbf{b}$. Therefore:
    \begin{equation}
        \max_{\boldsymbol{\theta}, \boldsymbol{\theta}'}\left(T(\rho(\boldsymbol{\theta}), \rho(\boldsymbol{\theta}'))\right) \leq \frac{n\epsilon}{\Tr(\mathcal{Q})},
    \end{equation}
and hence, using Lemma~\ref{lemma: security}, $\pi_H \mathcal{R} \approx_{n\epsilon/\Tr(\mathcal{Q})} \sigma_H \mathcal{S} \sigma_D$.

\textbf{Correctness:} The correctness proof is identical to the correctness proof of Thm~\ref{thm: luis defn composably private} (although with trace distance $n\epsilon/\Tr(\mathcal{Q})$ in the place of  $\sqrt{1 - \mathcal{P}^2(\mathcal{Q}, \mathbf{a})}$).
\end{proof}

As discussed in~\cite{hassani2025privacy}, the $\epsilon$ error in Eq.~\ref{eq: hassani proportionality} is equivalent to $\varepsilon$ in Eq.~\ref{eq: hamiltonian} up to a constant factor (which can hence be omitted if the distinguishing probability is normalised appropriately).


\section{State verification} \label{sec: state verification}

A central goal of composable definitions of security is being able to combine protocols, as is often necessary in realistic schemes. In this work, we have thus far considered that the concrete resource is a state $\rho$ which is distributed across the network. However, this may not be a reasonable assumption, and in order for parties to trust the security of the scheme, they must carry out state verification~\cite{shettell2022private, unnikrishnan2022verification, shettell2020graph}. 

In general, state verification schemes proceed by receiving a series of states from an untrusted source, and making measurements on these (possibly distributed across a network), potentially discarding some states, and keeping one target copy, which is the state to be used by the proceeding protocol. It is important that the order of the measurements, and in particular which copy is to be used as the target copy, is randomly chosen and kept hidden from the source. Generally these will be stabilisers of the state, and the overall results can be understood as a failure probability, based on the number of measurements which give a result of -1. The results of the measurements  give a confidence window on the quality of the state, typically of the form:
\begin{equation} \label{eq: verification confidence}
    \Pr(F(\rho,\rho^*) \geq 1- \lambda) \geq 1 - \delta,
\end{equation}
where $\rho$ is the intended state, $\rho^*$ is the target copy, and $\lambda, \delta$ may depend on variables including the size of the network, the number of copies used, and the failure rate.

Let us consider the parameter estimation scheme with the additional step that we use a state $\rho^*$ which satisfies Eq.~\ref{eq: verification confidence}, to some intended state $\rho$ which satisfies the condition of Eq.~\ref{eq: lemma}. Firstly, we note that, in the case $F(\rho,\rho^*) \geq 1- \lambda$:
\begin{equation}
    \max_{\boldsymbol{\theta}, \boldsymbol{\theta}'}(T(\rho^*(\boldsymbol{\theta})), \rho^*(\boldsymbol{\theta}')) \leq \varepsilon + \sqrt{\lambda},
\end{equation}
using the triangle inequality applied to the trace distance, the fact that the trace distance is non-increasing over CPTP maps, and the upper bound on the trace distance from the fidelity. From Lemma~\ref{lemma: security}, and following similar arguments to Sections~\ref{sec: luis defn} and~\ref{sec: majid defn} for the full constructive proof, it is straightforward to see that the protocol then has security $\varepsilon + \sqrt{\lambda}$.

The remaining issue is the probability $\delta$ that state verification fails. In this case, we must assume a worst-case scenario, that (following the abstract cryptography methods described previously) a distinguisher would have a distinguishing advantage between this implementation and an ideal protocol of 1. Thus, the overall security of the scheme following this state verification method would be $(1 - \delta)(\varepsilon + \sqrt{\lambda}) + \delta$. The most likely application would be that the intended state $\rho$ is a fully private state (such as a GHZ state in the case of mean estimation), in which case $\varepsilon = 0$.

Now let us assume that we are using states that are locally equivalent to graph states, and thus we can use the composable security of~\cite{colisson2024all}. This work presents abstractions of the state verification scheme (of the two, we will use the simpler resource $\mathcal{V}_{\ket{G}}$, by assuming that it is implemented through a scheme that includes the appropriate compilation step). We must consider the filtered ideal resource, $\mathcal{V}_{\ket{G}}$, where $\ket{G}$ is the desired graph state. The resource has $n+1$ interfaces, where $n$ of these are the interfaces of $P$ that comprise the network, and an interface $S$ corresponding to a source which distributes the quantum state. The protocol implemented by $\mathcal{V}_{\ket{G}}$ is to send the $i$-th qubit of the state $\ket{G}$ to the party $i$. It also includes an abort protocol, which is suppressed by adding the filter $\perp_{P \cup S}$.

The composition of this procedure with parameter estimation is straightforward. We will consider a modification of the concrete resource $\mathcal{R}$ for function estimation, Alg.~\ref{concrete resource general}, but with the first step (distributing the state $\rho$) omitted. Let us call this resource $\mathcal{R}'$. We see that $\mathcal{R}'$ is in fact equivalent to authenticated classical broadcasting across the network. We can compose this in sequence with $\mathcal{V}_{\ket{G}}$, to give the new resource $\mathcal{V}_{\ket{G}}\Vert\mathcal{R}'$. We add the additional functionality of the abort protocol, but as in~\cite{colisson2024all}, this can be omitted w.l.o.g. when considering a distinguisher with access to this functionality.

We would like to show that we can construct $\mathcal{R}$ exactly from $\mathcal{V}_{\ket{G}}\Vert\mathcal{R}'$ (using Defn.~\ref{defn: constructive cryptography}). Firstly, correctness: $ \pi_H \perp_{P \cup S} \mathcal{V}_{\ket{G}}\Vert\mathcal{R}' \approx \pi_H \mathcal{R} $, which immediately follows from our definition of the protocols -- that is, state verification with authenticated broadcast is functionally identical to the resource $\mathcal{R}$ that we considered previously. The security condition $\pi_H \mathcal{V}_{\ket{G}}\Vert\mathcal{R}' \approx \sigma_H \mathcal{R} \sigma_D$ is also straightforward: $\sigma_D$ can be omitted, $\sigma_H = \pi_H$, and we can also add a simulator $\sigma_S$ that receives a state from the source, although we do not need to use this.

 Therefore, if there is a state verification scheme consisting of resource $\mathcal{T}$ and protocol $\tau$ which can construct $\mathcal{V}_{\ket{G}}$ to within $\varepsilon'$, and given that $\mathcal{R}_\pi \xrightarrow{\varepsilon} \mathcal{S}_\Diamond$, where $\mathcal{S}$ represents parameter estimation and other symbols are defined as in Section~\ref{sec: general fn}, then by the triangle inequality, $\mathcal{T} \xrightarrow{\tau, \varepsilon + \varepsilon'} \mathcal{S}_{\Diamond}$ (see e.g.~\cite{portmann2014cryptographic} for further details).

\section{Discussion}

In this work, we have showed the composable security of using quantum networks to estimate functions of private parameters. We gave an explicit proof that the use of GHZ states to estimate the mean of a set of parameters is fully private, as well as showing that the privacy definitions considered in~\cite{bugalho2025private} and~\cite{hassani2025privacy} are both composably secure. Our work allows these two conditions to be compared, and applies to the further results regarding classes of private states considered in~\cite{bugalho2025private}.

Composable security is an important benchmark to allow protocols to be realised with confidence in realistic settings, allowing them to be repeated, and used alongside other protocols. A vital step in secure parameter estimation is state verification, which allows the users to participate in the scheme without having to put their full trust in the source. Using the composable security proof of~\cite{colisson2024all}, we showed how parameter estimation can be combined with state verification.

An obvious question is how to interpret the abstract definition of $\varepsilon$-secure in the context of parameter verification. Loosely speaking, this can be interpreted as the probability of a `failure' of the protocol (see the discussion in~\cite{portmann2014cryptographic}), however, it is not inherently defined in the framework of abstract cryptography how serious this failure is. Therefore, this can be analysed using the game-based paradigm, by considering how failure can occur and what the resulting information leakage would be, and then bounding the probability of this occurring using $\varepsilon$.

For example, given that all parties (including those in $D$) receive all measurement outcomes from the honest parties, there are certain pathological situations (such as if the source distributes the $\ket{+}$ state to all parties) in which this measurement outcome provides a bit of information about a local parameter $\theta_i$ (that is, the dishonest parties can construct estimators for the local parameters, with a variance that depends on the knowledge they can gain through protocol failures). Thus, in practical implementations, there should be an upper limit of the acceptable number of repetitions of the protocol, where the rate of information leaked about private parameters is balanced against the rate of information learned about the target function.

The combination of quantum metrology and cryptography is a burgeoning field, which has great promise in delivering near-term quantum advantage through initiatives such as the quantum internet. Through the framework considered here, we have explored the link between the quantum Fisher information, a standard tool in quantum metrology, and abstract cryptography, through the former's use as a measure of the accessible information of a state. Alongside the use case considered here, this may have more general interest -- given any state that is intended to communicate only information about a particular value, we have shown how the quantum Fisher information can be used as a measure of how that state can be constructed from knowledge of that value, and how this is the central principle of a constructive cryptography proof of security. 

\section{Acknowledgements}
\label{sec: acknowledgements}

Thank you to Léo Colisson, Raja Yehia, Luís Bugalho, and Majid Hassani for useful discussions and thoughtful comments. In particular thank you to LC for suggesting the use of filters which is used in the text, and help in understanding how to interpret early results, and to MH for proofreading the paper. We thank anonymous reviewers for their feedback. Thank you to Santiago Scheiner for his code template which was used to make Fig.~\ref{fig: parameter estimation}. Both authors acknowledge support from the Quantum Internet Alliance (QIA), which has received funding from the European Union’s Horizon 2020 research and innovation programme under grant agreement No 820445 and from the Horizon Europe grant agreements 101080128 and 101102140; the PEPR integrated project EPiQ ANR-22-PETQ-0007 and HQI ANR-22-PNCQ-0002, part of Plan France 2030; and the ANR project QNS ANR-24-CE97-0005-01.

\bibliography{references}
\bibliographystyle{unsrtnat}

\newpage
\appendix

\section{Quantum Fisher information and the Bures metric} \label{sec: QFI appendix}

For review articles relevant to this work, see for example~\cite{liu2020quantum, sidhu2020geometric, paris2009quantum}. Useful details are also given in~\cite{bugalho2025private, hassani2025privacy}.

Consider a random variable, $x$, which can take several different values as determined by a probability distribution, $f(x)$. The probability distribution can be parametrised by another variable $\theta$, hence we label it $f(x, \theta)$. For example (based on the example in Section~\ref{sec: mean estimation}), $x$ may be a binary variable which may take value 0 with probability $\frac{1}{2}(1 + \cos(\theta))$, and value 1 with probability $\frac{1}{2}(1 - \cos(\theta))$. Let's say that $x$ is a thing that we can measure, whereas $\theta$ is a variable that depends on some environmental factor. By sampling $x$, we learn about $f(x, \theta)$ and thus $\theta$.

A distribution that changes more dramatically when $\theta$ is changed will give more information about $\theta$ each time $x$ is sampled. For example, if $\theta$ is restricted to 0 or $\pi$, the distribution above will only require a single sample of $x$ to determine $\theta$. On the other hand, the distribution $g(x, \theta) = \frac{1}{2} + (-1)^x \frac{1}{10} \cos(\theta)$ will give less information about $\theta$ from each sample.

The Fisher information (for classical distributions) is a metric that quantifies the response of the distribution to the parameter:
\begin{equation}
    F(\theta) = \int \Pr(\mathbf{x} | \theta) \left(\partial_\theta \ln[\Pr(\mathbf{x}|\theta)]\right)^2 d\mathbf{x}.
\end{equation}
This can be used to bound the precision with which it is possible to estimate $\theta$ based on $m$ measurements of $x$, which asymptotically approaches the Cramér-Rao bound:
\begin{equation}
    \Delta \theta^2 \geq \frac{1}{mF(\theta)}.
\end{equation}

If the distribution depends on $n$ parameters, $\boldsymbol{\theta}$, we instead represent the information by the $n \times n$ Fisher information matrix, which also contains information about how different parameters are related to each other -- if parameters $\theta_i$ and $\theta_j$ are completely uncorrelated, $F_{ij} = 0$.

When using a quantum process to extract information about a particular parameter from a state, the type of measurement which is used to create the distribution is important, as well as the dynamics used to encode the parameter into a particular probe state. If we assume a particular measurement, this will give rise to a particular distribution of outcomes. Using this, the quantum analogue of the Fisher information can be constructed directly from the classical Fisher information:
\begin{equation}
    \mathcal{I}(\Pr(\mathbf{x}|\theta)) = \int \frac{1}{\Tr[\Pi_\mathbf{x} \rho(\theta)]} \left(\partial_\theta \Tr[\Pi_\mathbf{x} \rho(\theta)]\right) d\mathbf{x},
\end{equation}
which is in terms of a positive-operator valued measurement $\{ \Pi_\mathbf{x} \}_\mathbf{x}$ made onto the encoded state $\rho(\theta)$.

In order to find the amount of information that can theoretically be extracted from the state, we need to optimise over all possible measurements, which leads to the quantum Fisher information:
\begin{equation}
    \mathcal{Q}(\theta) = \Tr(\rho(\theta) L_\theta^2),
\end{equation}
where $L_\theta$ is an operator defined so that:
\begin{equation}
    \partial_\theta \rho(\theta) = \frac{\rho(\theta) L_\theta + L_\theta \rho(\theta) }{2}.
\end{equation}
The quantum Fisher information matrix (QFIm), which is defined regarding the encoded state $\rho(\boldsymbol{\theta})$, is:
\begin{equation}
    \mathcal{Q}_{\mu \nu} (\boldsymbol{\theta}) = \frac{1}{2}\Tr(\rho(\boldsymbol{\theta})(L_{\theta_\mu} L_{\theta_\nu} + L_{\theta_\nu} L_{\theta_\mu})),
\end{equation}
from which we have the quantum Cramér-Rao bound, which bounds the covariance matrix of the estimation of the set of parameters $\boldsymbol{\theta}$:
\begin{equation}
    \text{cov}(\boldsymbol{\theta}) \geq \frac{1}{m} \mathcal{Q}^{-1}(\boldsymbol{\theta}),
\end{equation}
where $m$ is once again the number of measurements.

The ability to achieve the Cramér-Rao bound depends on the measurements made, and this may indeed not be attainable, either due to restrictions on the scheme (for example, if entangling measurements are required), or if the measurement itself depends on $\boldsymbol{\theta}$, which we assume is unknown by the party carrying out the estimation. This is not necessarily the case, and often in multi-parameter estimation the Cramér-Rao bound can be achieved, which gives a quadratic advantage in the efficiency of quantum measurement schemes over classical schemes~\cite{proctor2018multiparameter}.

We can also reparametrise the QFIm: if we write $\mathcal{Q}$ in terms of a rotated set of parameters $\boldsymbol{\theta}'$ which can be expressed in terms of the original parameters $\boldsymbol{\theta}$, then the corresponding QFIm is given by:
\begin{equation}
    \mathcal{Q}(\boldsymbol{\theta}') = B^{\text{T}} \mathcal{Q}(\boldsymbol{\theta}) B,
\end{equation}
where $B_{\mu \nu} = \partial \theta_\mu / \partial \theta'_\nu$.

One interpretation of the QFI is a measure of how the quantum state responds to changes in $\boldsymbol{\theta}$ -- roughly speaking, a greater change in the state requires fewer measurements to notice, and therefore reveals more information about the parameter. This is clear with the connection to the Bures metric~\cite{paris2009quantum,vsafranek2017discontinuities}, which is defined to represent the distances of two states on the quantum manifold under infinitesimal changes of their parameters:
\begin{equation}
    D^2_B (\rho(\boldsymbol{\theta}), \rho(\boldsymbol{\theta} + d\boldsymbol{\theta})) = g_{\mu \nu} d\theta_\mu d\theta_\nu,
\end{equation}
where we have used the (direct) Bures distance:
\begin{equation}
    D^2_B (\rho, \sigma) = 2\left(1 - \sqrt{F(\rho, \sigma)}\right).
\end{equation}
The QFIm is proportional to the Bures metric:
\begin{equation}
    \mathcal{Q}_{\mu \nu} (\rho (\boldsymbol{\theta})) = 4 g_{\mu \nu}.
\end{equation}
To find the Bures distance between two states, it is necessary to integrate the Bures metric along the geodesic between them~\cite{spehner2025bures,spehner2014quantum}.

\end{document}